\documentclass[conference]{IEEEtran}

\usepackage{epsf}
\usepackage{times}
\usepackage{epsfig}
\usepackage{graphicx}
\usepackage{amsmath}
\usepackage{amssymb}
\usepackage{amsxtra}
\usepackage{here}
\usepackage{rawfonts}
\usepackage{times}
\usepackage{url}
\usepackage{cite}
\usepackage{amsmath}
\usepackage[ruled,vlined]{algorithm2e}
\usepackage{epstopdf}
\usepackage{amsmath,amsthm}
\usepackage{subfigure}
\usepackage{multirow}

\newtheorem{lemma}{\bf Lemma}

\newlength{\aligntop}
\newlength{\alignbot}
\makeatletter
\renewenvironment{align}{%
  \vspace{\aligntop}
  \start@align\@ne\st@rredfalse\m@ne
}{%
  \math@cr \black@\totwidth@
  \egroup
  \ifingather@
    \restorealignstate@
    \egroup
    \nonumber
    \ifnum0=`{\fi\iffalse}\fi
  \else
    $$%
  \fi
  \ignorespacesafterend%
  \vspace{\alignbot}\par\noindent
}

\IEEEoverridecommandlockouts

\begin{document}
\title{\huge Single Controller Stochastic Games for\\Optimized Moving Target Defense\vspace{-0.7cm}}
\author{\authorblockN{ AbdelRahman Eldosouky$^1$, Walid Saad$^1$, and Dusit Niyato$^2$} \authorblockA{\small
$^1 $Wireless@VT, Bradley Department of Electrical and Computer Engineering, Virginia Tech, Blacksburg, VA, USA, Emails:\{iv727,walids\}@vt.edu \\
$^2$ School of Computer Engineering, Nanyang Technological University (NTU), Singapore, Email:\ {dniyato@ntu.edu.sg}\\
\vspace{-1cm}
}%
 \thanks{This research was supported by the U.S. National Science Foundation under Grants ACI-1541105, CNS-1524634, and CNS-1446621.}  
 } 
\date{}
\maketitle

\begin{abstract}
Moving target defense (MTD) techniques that enable a system to randomize its configuration to thwart prospective attacks are an effective security solution for tomorrow's wireless networks. However, there is a lack of analytical techniques that enable one to quantify the benefits and tradeoffs of MTDs. In this paper, a novel approach for implementing MTD techniques that can be used to randomize cryptographic techniques and keys in wireless networks is proposed. In particular, the problem is formulated as a stochastic game in which a base station (BS), acting as a defender seeks to strategically change its cryptographic techniques and keys in an effort to deter an attacker that is trying to eavesdrop on the data. The game is shown to exhibit a single-controller property in which only one player, the defender, controls the state of the game. For this game, the existence and properties of the Nash equilibrium are studied, in the presence of a defense cost for using MTD. Then, a practical algorithm for deriving the equilibrium MTD strategies is derived. Simulation results show that the proposed game-theoretic MTD framework can significantly improve the overall utility of the defender, while enabling effective randomization over cryptographic techniques.
\end{abstract}

\section{Introduction}

The emergence of reconfigurable wireless networks based on software defined networking and software defined radio concepts is expected to revolutionize the future of wireless communications. However, such reconfigurable systems are susceptible to many security threats that range from jamming to eavesdropping and node forgery. One effective way to thwart attacks on reconfigurable environments is via the use of \emph{moving target defense (MTD)} techniques \cite{MTD07}. MTDs are built on the premise of continuously randomizing the network's configuration (e.g., cryptographic keys, network parameters, IP addresses) so as to increase the uncertainty and cost of attack on the adversary. The effective deployment of MTDs requires meeting several challenges that range from optimizing the randomization to analyzing the costs and benefits of MTDs \cite{MTD08,MTD15,MTD12,MTD02,MTD13,MTD11,MTD10,MTD09,MTD14}.

A number of research works have recently attempted to address some of these challenges \cite{MTD07,MTD08,MTD15,MTD12}. First, in \cite{MTD07},  , the authors focus on the five dominant domains in which MTD techniques could be applied against cyber attacks in critical systems. In this work, defined these domains to be networks, platforms, runtime environments, software, and data. They studied the weakness and advantages of using MTD in these domains. In \cite{MTD08}, the authors proposed a three-layer model to evaluate the effectiveness of MTDs in software. These layers capture low-level contexts in separate programs, model damage propagation between different programs, and provide a user interface to expresses evaluation results. The work in \cite{MTD15} considers an MTD to be a subclass of system agility. In this work, system agility is defined as any reasoned modification to a system or environment in response to a functional, performance, or security need. In \cite{MTD12}, the authors propose a foundation for defining the theory of MTD. They defined key problems and hypothesis related to MTD such as the way to select the next valid configuration of the system, configuration space, and the timing problem.

The use of MTD in resource-constrained distributed devices, e.g., wireless sensor networks was studied in \cite{MTD02}. The authors proposed two different reconfigurations at different architectural layers. The first is applied at what they defined as a security layer by using a number of cryptographic techniques and each node in the network can choose its encryption method for each packet by adding a special identifier in the packet header. The second reconfiguration is to be applied at the physical layer by changing the node's firmware.
In \cite{MTD13}, the authors use MTD to defend against selective jamming attacks. This work studies the problem of isolating a subset of the network by jamming the signals sent from this sub-network. The work in \cite{MTD13} also provides practical MTD solutions such as address flipping and random address assignment.
The use of software defined networking in applying MTD was discussed in \cite{MTD11}. The authors defined a technique to MTD by assigning virtual IPs to hosts in the network beside their reals IPs. Software-defined networking was used to manage the IP translation. However, these works are mostly qualitative or experiment based and, as such, they do not address specific MTD problem formulations.

More recently, game-theoretic methods have recently attracted attention as a suitable tool for implementing MTDs \cite{MTD10,MTD09,MTD14}. In \cite{MTD09}, the authors develop a zero-sum stochastic game model to a feedback-driven multi-stage MTD. A feedback learning framework was used to implement MTD based on real-time data and observations made by the system. The purpose of the learning algorithm for the defender is to monitor its current state and update its randomized strategy based on its observation. In this model, the attacker launches a multi-stage attack and the defender responses at each layer. In \cite{MTD14}, the authors analyze a system in which the defender has a number of different platforms to run a critical application and the attacker has a set of attacks that are applicable against some of these platforms. The authors proposed two types of attackers, static and adaptive, and gave attack model to both of them. The authors in \cite{MTD15}, also suggested that MTD games should be modeled as tunable hierarchical games. The output of a game at one level should determine the level of risk associated with a game at a different level.

A recent collection of publications for applying game theory in MTD \cite{MTD10} does not provide clear approaches to concretely reach equilibrium strategies. Moreover, the works in \cite{MTD09} and \cite{MTD14} abstract many of the details of the network considered, and, thus, they cannot directly apply to practical systems. Moreover, the work in [9] assumes the presence of a highly intelligent layered attacker which may not be true in practice.

The main contribution of this paper is to develop a novel game-theoretic model for MTD that can be applied to securing a wireless network. In particular, we consider a wireless system in which a base station (BS) seeks to implement an MTD-based cryptographic approach in which it randomizes over various cryptographic keys and techniques so as to evade an eavesdropper that is trying to decrypt the messages. We formulate the problem as a single-controller non-zero-sum stochastic game in which the BS uses a number of cryptography techniques along with a number of keys for each technique. The BS can implement MTD by randomizing over various actions that include choosing an encryption method defined by specific encryption technique and key combination. We also consider a defense cost for applying MTD that depends on the number of consecutive changes in the system.
Since our model deals with resource-constrained systems, the encryption techniques should not be highly resource consuming. Therefore, we develop an approach that attempts to avoid the use of encryption techniques with long encryption keys in order to decrease the power consumption. While short-key encryption techniques are more vulnerable against attacks, MTD will allow the BS to switch between encryption techniques and so it is unlikely that the attacker will be able to reveal the key before it is changed.
For this game, it is shown that a Nash equilibrium always exists. To find this equilibrium, we propose an algorithm based on bimatrix game equilibrium defined for all possible pure stationary strategies of the original game. Simulation results show that the proposed approach will yield a higher defender's utility when compared with other schemes that randomly pick the strategies. 

The rest of the paper is organized as follows.
Section~\ref{sec:sysmodel} provides the system model, assumptions, and defines the defender's and attacker's utilities. In Section~\ref{sec:syssol}, the stochastic game is formulated and the steps of calculating equilibrium points are shown, and also a way to define cost function in MTD. Simulation results are discussed in Section~\ref{sec:sim}. Finally, conclusions are drawn in Section~\ref{sec:conclusion}.

\section{System Model and Problem Formulation}\label{sec:sysmodel}
Consider a wireless sensor network that consists of a BS and a number of wireless nodes. The network is deployed for sensing and collecting data about some phenomena in a given geographic area. Sensors will collect data and use multi-hop transmissions to forward this data to a central receiver or BS. The multiple access follows a slotted Aloha protocol.  Time is divided into slots and the time slot size equals the time required to process and send one packet. Sensor nodes are synchronized with respect to time slots. We assume that nodes are continuously working and so every time slot there will be data that must be sent to the BS.

All packets sent over the network are assumed to be decrypted using a given encryption technique and a previously shared secret key. All the nodes in the system are pre-programmed with a number of encryption techniques along with a number of encryption keys per technique, as what is typically done in sensor networks~\cite{MTD02}. The BS chooses a specific encryption technique and key by sending a specific control signal over the network including the combination it wants to use. We note that the encryption technique and key sizes should be carefully selected in order not to consume a significant amount of energy when encrypting or decrypting packets. Increasing the key size will increase the amount of consumed energy particularly during the decryption ~\cite{MTD01}. Since the BS is mostly receiving data, it spends more time decrypting packets rather than encrypting them and, thus, it will be highly affected by key size selection.

In our model, an eavesdropper is located in the communication field of the BS and it can listen to packets sent or received by the BS. As packets are encrypted, the attacker will seek to decrypt the packets it receives in order to get information. The attacker knows the encryption techniques used in the network and so it can try every possible key on the received packets until getting useful information. This technique is known as brute-force attack.

The idea of using multiple encryption techniques was introduced in~\cite{MTD02}. However, in this work, each node individually selects one of these technique to encrypt transmitted packets. The receiving node can know the used technique by a specific field in the packet header. Large encryption keys were used which require a significant amount of power to be decrypted. Nonetheless, these large keys are highly unlikely to be revealed using a brute-force attack in a reasonable time. Here, we propose to use small encryption keys to save energy and, in conjunction with that, we enable the BS to change the encryption method in a way that reduces the chance that the encryption key is revealed by the attacker. This is the main idea behind MTD. In MTD techniques, the defender aims to change the attack surface~\cite{MTD06} which represents the points that could be attacked. In this model, the encryption key represents the attack surface, and by changing the encryption method, the BS will make it harder for the eavesdropper to reveal the key and get the information from the system.

Naturally, the goals of the eavesdropper and the BS are not aligned. On the one hand, the BS wants to protect the data sent over the network by changing encryption method. On the other hand, the attacker wants to reveal the used key in order to get information. To understand the interactions between the defender and the attacker, one can use game theory to study their behavior in this MTD scenario. The problem is modeled as a game in which the attacker and the defender are the players. As the encryption method should be changed over time and depending on the attacker's actions, we must use a dynamic game.

Thus, we formulate a stochastic game $\mathit{\Xi}$ described by the tuple $\left\langle \mathcal{N}, \mathcal{S}, \mathcal{A}, \mathcal{P}, \mathcal{U}, \beta \right\rangle$ where $\mathcal{N}$ is the set of the two players: the defender $p_1$, the BS, and the attacker $p_2$, the eavesdropper. $\mathcal{S}$ is the set of game states and $\mathcal{A}$ is the set of actions defined for each player at every state. $\mathcal{P}$ is the set of transition probabilities between states. $\mathcal{U}$ is the set of utilities each player will get for a given combination of actions and state. Finally, $0 < \beta < 1$ is a discount factor.

The defender can choose to use one of the $N$ available encryption techniques or to use the current encryption technique with one of the $M$ available encryption keys predefined for this technique. Each game state is well defined by the current encryption technique and key combination. Therefore, there will be $K=N \cdot M$ states, i.e., $\mathcal{S}=\{s_1,s_2,\dots,s_{K}\}$. In each state $s \in \mathcal{S}$, each player has a set of actions $\mathcal{A}_i$. Let $\mathcal{A}_1 = \{a^1_1,a^1_2, \dots, a^1_K\}$ be the defender's actions which represent the choice of a specific technique and key combination among the available $K$ combinations. Let $\mathcal{A}_2 = \{a^2_1, \dots, a^2_N\}$  be the action set of the attacker which represents the set of techniques that the attacker is trying to decrypt.

In each state $s \in \mathcal{S}$ and for each action pair in $\mathcal{A}_1 \times \mathcal{A}_2$, there is an outcome (payoff) for each player. This outcome depends on the current state and actions taken by both players in this state. This outcome is defined by player-specific utility functions in $\mathcal{U}$.
For given actions $a^1 \in \mathcal{A}_1$ and $a^2 \in \mathcal{A}_2$, the defender's utility at state $s_i$ is given by:
\vspace{-0.05cm}
\begin{equation}
U_1(a^1,a^2,s_i) = R_1(a^2)+T_1(a^1,a^2,s_i)-P_1(s_i), 
\vspace{-0.15cm}
\end{equation}
where $R_1$ is the reward gained from protecting a packet. This reward depends on the attacker's action as the defender will obtain a  higher reward if the eavesdropper is attacking another encryption technique. $P_1$ is the power used to decrypt a packet and it depends on the technique (state). $T_1$ is the transition reward that the defender will gain from applying MTD and choosing a key-technique combination. This reward depends on the current system state, the defender's action taken at this state (which determines the next state), and attacker's action.

Similarly, the attacker's utility at state $s_i$ for given actions $a^1 \in \mathcal{A}_1$ and $a^2 \in \mathcal{A}_2$ will be given by:
\vspace{-0.05cm}
\begin{equation}
U_2(a^1,a^2,s_i) = R_2(a^1,a^2,s_i) - P_2(s_i),
\vspace{-0.15cm}
\end{equation}
where $R_2$ is the attacker's reward from examining the encryption keys for a given technique. Here, if the attacker can examine more keys, it will get closer to revealing the actual key. This reward depends on the attacker's action, current encryption technique (state), and defender's action. $P_2$ is the power used to decrypt a packet that depends also on the current technique.

Based on these rewards, the game is non-zero sum. Thus means, every player will try to maximize its reward and the sum of rewards is not zero. This stochastic game also exhibits an interesting property pertaining to the fact that the transition probabilities in $\mathcal{P}$ depend only on the actions of the defender. Moreover, when the defender selects an action at one state, the game moves to another state defined by the encryption technique and key combinations with a probability $p=1$. This type of stochastic games is known as \emph{single-controller stochastic games} ~\cite{MTD16}. 

This type of games is most suitable for MTD problems in which the defender aims at randomizing system parameters, as the goal of MTD is to change system parameters in order to harden the attacker's mission. The defender should take actions to change these parameters within a reasonable time. Single-controller stochastic games satisfy this property by allowing the defender to control the actions thus changing the game state which maps to changing system parameters in MTD.

\section{Proposed MTD Game Solution}\label{sec:syssol}

\subsection{Equilibrium Strategy Determination}

The studied game is a finite stochastic games since the number of states and the number of actions per state are finite.
Stochastic games are dynamic in the sense that the game moves between states each time step. In stochastic games we are interested in the accumulated (total) utilities of the players over time. Discounted utilities over time are typically used by summing the current utility and all the expected future utilities multiplied by a discount factor. In such cases, players are interested more in current payoffs than future ones. Each player seeks to take actions that maximize its utility given the other player's actions. When no player can improve its utility by solely changing its actions, the game is said to be at equilibrium.

For discounted stochastic games, the existence of Nash equilibrium points in stationary strategies was proven ~\cite{MTD03}. Stationary strategies are those strategies in which the actions taken at each state depend on this state only. If at each state, the player selects a specific action with probability $p=1$ then this called pure stationary strategy. If the player chooses between actions with some probabilities then it is called a \emph{mixed stationary strategy}.

In~\cite{MTD04}, the authors propose a scheme that can find a Nash equilibrium point for discounted non-zero sum single-controller stochastic games. The key idea is to form a bimatrix game (one matrix for each player). The rows and columns of each matrix represent pure stationary strategies for each player. The elements of these matrices represent the accumulated discounted utilities over all states (recursion) for every strategies pair. Then, any mixed strategy Nash equilibrium of this bimatrix game can be used to get a Nash equilibrium of the stochastic game.

Since the defender is the controller which selects actions to move the game to a specific state, time steps of the stochastic game are controlled by the defender. Assuming that the attacker has enough power, it can complete the brute-force attack in time $t_i$ for $i=1,2,\dots,N$ for each one of the encryption techniques. Then, the defender should choose the time step $t$ to take the next action as follows:
\begin{equation}
t<\min(t_i), \ \  i=1,2,\dots,N.
\end{equation}
By doing this, the defender can make sure that it takes a timely action before the attacker succeeds in revealing one of the keys.

The accumulated utility of player $i$ at state $s$ will be:
\begin{equation}\label{eq:phi}
\Phi_{i}(\boldsymbol{f},\boldsymbol{g},s)=\sum_{t=1}^{\infty} \beta^{t-1} \cdot  U_i(f(s_t),g(s_t),s_t),
\end{equation}
where $\boldsymbol{f}$ and $\boldsymbol{g}$ are the strategies adopted by the defender and attacker, respectively. The strategy specifies a vector of actions to be chosen at each of the states, e.g., $\boldsymbol{f}=[f(s_1),\dots,f(s_K)]$ for all the $K$ states. Actions $f(s_t)$ and $g(s_t)$ are the actions chosen at $s_t$, which is the state of the game at time $t$, according to strategies $\boldsymbol{f},\boldsymbol{g}$. State $s_t \in \mathcal{S}$ is determined by the defender's action at time $t-1$. The game is assumed to start at a specific state $s=s_1$. Note that the utility in (\ref{eq:phi}) is always bounded at infinity due to the fact that $0<\beta<1$.

When designing the bimatrix, the defender needs to calculate the accumulated utility when choosing each pure strategy against all of the attacker's pure strategies. The defender, as a controller, can know the next state resulting from its actions, and, thus, it sums the utilities in all states using the discount factor $\beta$. Let $\boldsymbol{X}$ be the defender's accumulated utility matrix for all defender's pure strategies' permutations and all attacker's pure strategies' permutations. We let $\boldsymbol{F}_{i\boldsymbol{.}}=[\boldsymbol{f}_1,\boldsymbol{f}_2,\dots,\boldsymbol{f}_{K^K}]$ be a matrix of all defender's pure strategies' permutation where each row represents actions in this strategy and similarly $\boldsymbol{G}_{i\boldsymbol{.}}=[\boldsymbol{g}_1,\boldsymbol{g}_2,\dots,\boldsymbol{g}_{N^K}]$ the matrix of all attacker's pure strategies' permutation. Then each element $X_{i,j}$ of $\boldsymbol{X}$ will be given by:
\begin{equation}\label{eq:X}
X_{i,j} =\sum_{\mathcal{S}}  \Phi_{1}(\boldsymbol{F}_{i\boldsymbol{.}},\boldsymbol{G}_{j\boldsymbol{.}},s), \forall i, j,
\vspace{-0.25cm}
\end{equation}
where $i=1,\cdots,K^K$ and $j=1,\cdots,N^K$. The attacker can only calculate its payoffs at time $t=1$, as the attacker cannot know in advance the actions taken at each state and hence the reward it will get in future. Similarly, let $\boldsymbol{Y}$ be the attacker's accumulated utility matrix, then each element $Y_{i,j}$ of $\boldsymbol{Y}$ will:
\begin{equation}
Y_{i,j} =\sum_{\mathcal{S}}  \Phi_{2}(\boldsymbol{F}_{i\boldsymbol{.}},\boldsymbol{G}_{j\boldsymbol{.}},s), \forall i, j, 
\vspace{-0.2cm}
\end{equation}
where $i$ and $j$ are the same as the defender's case, and $\Phi_{2}(\boldsymbol{F}_{i\cdot},\boldsymbol{G}_{j\cdot},s)$ is only evaluated at time $t=1$.

The solution of the bimatrix could be obtained by algorithms such as Lemke-Howson~\cite{MTD05}, which is proven to always terminate at a solution and hence finds a mixed Nash equilibrium of the bimatrix game. This solution is then used as in~\cite{MTD04} to find the equilibrium of the stochastic game. Let $(\boldsymbol{x}^*, \boldsymbol{y}^*)$ be any mixed strategy Nash equilibrium point for the bimatrix game $(\boldsymbol{X}, \boldsymbol{Y})$. Each $(\boldsymbol{x}^*, \boldsymbol{y}^*)$ is a vector of probabilities with which each player can choose each strategy in all the strategies permutations.

As each strategy represents the set of actions per all states, the equilibrium point to the stochastic game, i.e, the probability of choosing each strategy, can be calculated as:
\begin{align}\label{eq7}
E^*_{i,j}=   \sum_{l=1,i=F_{l,j}}^{K^K} x_l^* , \ i = 1,\dots,K, j=1,\dots,K,  \nonumber \\
H^*_{i,j}= \sum_{l=1,i=G_{l,j}}^{K^K} y_l^*, \ i=1,\dots,N, j=1,\dots,K,
\end{align}

\hspace{-0.38cm}where $x_l^* \in \boldsymbol{x}^*$ and $y_l^* \in \boldsymbol{y}^*$ are the elements of $\boldsymbol{x}^*,\boldsymbol{y}^*$ that represent strategies' probabilities.
Each element $E^*_{i,j}$ of $\boldsymbol{E}^*$ and $H^*_{i,j}$ of $\boldsymbol{H}^*$ is the probability of taking action $i$ in state $j$ for the defender and the attacker, respectively.
The summations in (\ref{eq7}) give the probabilities of one action $i$ which satisfies the condition. This is repeated for all values of $i$ to get a column which is all actions' probabilities in one state. Different values of $j$ give the rest of the states.
$\boldsymbol{E}^*$ is a $K \cdot K$ matrix that gives the probability of each of the defender's $K$ actions in each of the $K$ states. Similarly, $\boldsymbol{H}^*$ is an $N \cdot K$ matrix that gives the probability of each of the attacker's $N$ actions in each of the $K$ states. These matrices are the \emph{equilibrium strategies} for both players.


These probabilities specify the behavior of the game. The defender in each state will choose an action (selecting an encryption method) with some probability and so the game will move to another state (encryption method). Then, again in the new state, the defender chooses a new action and so on. Using this process, the defender will keep moving between encryption methods which effectively implements a highly randomized MTD.

Finally, the value (expected utility) of each player at equilibrium can be computed by applying the equilibrium strategies and finding the accumulated payoffs of both players. These expected utilities are calculated by following all the possible transitions due to defender's actions in each state. Let $v^*_i(s)$ be player's $i$ value at state $s$:
\begin{equation}\label{eq:v}
v^*_i(s) = \Phi_{i}(\boldsymbol{E}^*,\boldsymbol{H}^*,s) \ \ s \in \mathcal{S},
\end{equation}

As the players get these values at equilibrium, both players will not have an incentive to deviate from these equilibrium strategies. The player who deviates will get a lower value when the other player uses its equilibrium strategy. This can be expressed as:
\begin{align}
v^*_1(s) &\geq  \Phi_1 (\hat{\boldsymbol{E}},\boldsymbol{H}^*,s), \  s \in \mathcal{S},\nonumber \\
v^*_2(s) &\geq \Phi_2(\boldsymbol{E}^*,\hat{\boldsymbol{H}},s), \ s \in \mathcal{S},
\end{align}

\hspace{-0.38cm}for any $\hat{\boldsymbol{E}}$ and $\hat{\boldsymbol{H}}$ other than the equilibrium strategies.

\subsection{Moving Target Defense Cost}

In previous sections, the defender's utility included a reward from applying MTD which corresponds to the gain from randomizing system parameters. However, applying MTD may incur associated costs for the defender. Examples include the cost of reconfiguring the system and changing parameters. In our decryption model, the BS might not be able to change the encryption method unless it ensures that all nodes are informed by the change, which requires some propagation time. Changing the method before this time can lead to a conflict in the used method between various nodes around the BS (e.g., nearby and far away nodes). 

We model this cost as a function of the number of consecutive encryption method changes in the past time steps. Let the number of consecutive method changes during the past time steps be $n$ and the cost value be $q$. The cost function will be $C(q,n)$ and it is an increasing function in the number of consecutive changes $n$. The defender's utility can then be written as:
\begin{equation}
U_1(a^1,a^2,s_i) = R_1(a^2)+T_1(a^1,a^2,s_i)-P_1(s_i)-C(q,n).
\end{equation}
Clearly, $n$ will be zero at the first time step. The effect of this cost can appear in the accumulated utility in (\ref{eq:phi}) which will affect the matrix $\boldsymbol{X}$ in (\ref{eq:X}) and the defender's equilibrium values in (\ref{eq:v}).

We propose two different functions to express the cost. The first cost function can be expressed as $C(q,n)=q \cdot n$. We need to make sure that the game will remain finite after adding this cost function so that the same solution can be applied. As the cost affects the utility, we can state the following lemma:

\begin{lemma}\label{lem1}
The accumulated defender's utility will remain bounded after adding a cost function in the form $C(q,n)=q \cdot n$ and, thus, the game will still admit an equilibrium point.
\end{lemma}

\begin{proof} We prove this lemma by rewriting the defender's accumulated utility:
\vspace{-0.5cm}

\begin{equation*}
\Phi_{1}(f,g,s)=\sum_{t=1}^{\infty} \beta^{t-1} \big( R_1(a^2)+T_1(a^1,a^2,s_i)-P_1(s_i)- q \cdot n \big).\vspace{-.1cm}
\end{equation*}
by noticing that the maximum for $n$ is $t-1$ and taking the limit as $t$ reaches $\infty$ we get
\begin{equation*}
\lim_{t\to\infty} \beta^{t-1} \big( R_1(a^2)+T_1(a^1,a^2,s_i)-P_1(s_i)- q \cdot (t-1) \big) = 0.\vspace{-0.2cm}
\end{equation*}
\vspace{-0.2cm}
\end{proof}

A second form for the cost function is $C(q,n)=q \cdot \ln(n+1)$. We choose such a logarithmic function to reduce the effect of cost propagation. Note, Logarithmic function has a smaller rate of growth compared to the linear function in the first case. We need to ensure that the game will remain finite by adding this cost function, so we state the following lemma:

\begin{lemma}\label{lem2}
The accumulated defender's utility will remain bounded after adding the cost function $C(q,n)=q \cdot \ln(n+1)$ and, thus, the game will still admit an equilibrium point.
\end{lemma}

\begin{proof} We prove this lemma in a manner analogous to Lemma ~\ref{lem1} where the limit will be zero also.
\vspace{-0.2cm}
\end{proof}

In general, any function could be used to represent the propagation cost when its limit is bounded at infinity.

\section{Simulation Results and Analysis}\label{sec:sim}

For our simulations, we choose a system that uses $2$ encryption techniques with $2$ different keys per technique. Thus, the number of system states are $4$ and the defender has $4$ actions in each state. For the bimatrix, the attacker has $2^4=16$ different strategy permutations and the defender has $4^4=256$ different strategy permutations. The power values are set to $1$ and $3$ to pertain to the ratio between the power consumption in the two different encryption techniques. These values are the same for both players. We set $R_1$ and $R_2$ to be $10$ and $5$ depending on the opponent's actions. We choose these values to be higher than the power values in order for the utilities to be positive. The transition reward is set to $5$ and $10$ for switching to another state defined by another key or another technique, respectively.

\begin{table}[t]
\centering
\caption{Attacker's and defender's equilibrium strategies}
\label{Tab1}
\begin{tabular}{c|c|c||c|c|c|c|}
\cline{2-7}
                         & \multicolumn{2}{c||}{Attacker} & \multicolumn{4}{c|}{Defender} \\ \cline{2-7} 
                         & $a_1 $            & $a_2$            & $a_1$    & $a_2$    & $a_3$    & $a_4$    \\ \hline
\multicolumn{1}{|c|}{$s_1$} & $0.7436$       & $0.2564$         & $0.0000$ & $0.6622$ & $0.1681$ & $0.1697$     \\ \hline
\multicolumn{1}{|c|}{$s_2$} & $0.7436$       & $0.2564$         & $0.4441$  & $0.0195$ & $0.1697$ & $0.3667$     \\ \hline
\multicolumn{1}{|c|}{$s_3$} & $0.3482$       & $0.6518$         & $0.4441$  & $0.3667$ & $0.0195$ & $0.1697$     \\ \hline
\multicolumn{1}{|c|}{$s_4$} & $0.3482$       & $0.6518$         & $0.4441$  & $0.3667$ & $0.1697$ & $0.0195$     \\ \hline
\end{tabular}
\vspace{-0.2cm}
\end{table}

First, we run simulations when there is no transition cost, $q=0$. The equilibrium strategies for both the attacker and defender are shown in Table~\ref{Tab1}. Note that actions $a_1,a_2$ represent the selection of two keys for the same encryption technique and actions $a_3,a_4$ represent two keys for another technique. Table 1 shows the probabilities over all actions for each player. These probabilities show how players should select actions in every state. For the defender, if it starts in state $s_3$ then it should move to state $s_1$ with the highest probability and move to state $s_2$ with a very similar probability. This is because the defender will change the technique and so gets a higher transition reward. We can see that the probability of moving to the same state is always very low and can reach $0$ as in state $s_1$. The probability of moving to a state that has a similar encryption key is always less than that of moving to a state with different technique as the transition reward will be lower. For the attacker, the probability of attacking the same technique that is used in the current state is always higher than attacking any other technique.

In Fig.~\ref{fig:1}, we show the effect of the discount factor on the defender's utility at equilibrium in every state. First, we can see that all utility values at all states increase as the discount factor increases. This is due to the fact that increasing the discount factor will make the defender care more about future rewards thus choosing the actions that will increase these future rewards. Fig.~\ref{fig:1} also shows that the defender's values at states $1$ and $2$ are higher than at states $3$ and $4$. This because states $1$ and $2$ adopt the first encryption technique which uses less power than the encryption technique used in states $3$ and $4$. The difference mainly arises in the first state before switching to other states and applying the discount factor. Clearly, changing the discount factor has a big effect on changing the equilibrium strategy, and, thus, the game will move between states with different probabilities resulting in a different accumulated reward.

\begin{figure}[t]
  \centerline{\includegraphics[width=7cm]{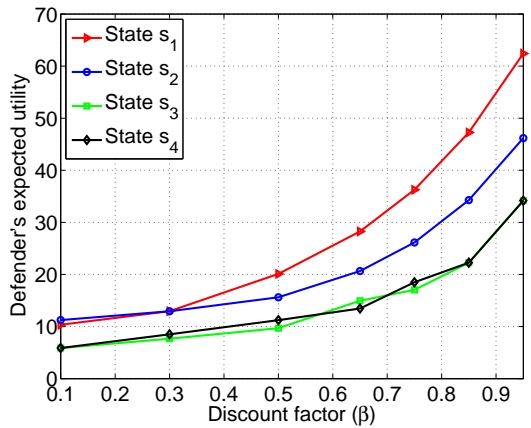}}
  \caption{The defender's expected utility in each state against discount factor $\beta$.}\label{fig:1}
    \vspace{-0.45cm}
\end{figure}

\begin{figure}[t]
  \centerline{\includegraphics[width=7cm]{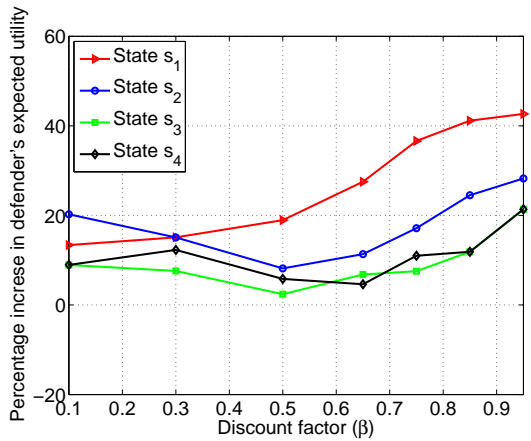}}
  \caption{Percentage increase in the defender's expected utility when using the equilibrium strategy and when using equal probabilities over actions. This is shown in each state as function of the discount factor $\beta$.}\label{fig:2}
    \vspace{-0.45cm}
\end{figure}

In Fig.~\ref{fig:2}, we study the effect of applying the proposed MTD technique against the case when the defender decides to use equal probabilities over its actions in each state, i.e., all entries equal $0.25$ as there are four actions per state. Fig.~\ref{fig:2} shows the percentage of increase in the defender's expected utility. We can see that the minimum increase is non-zero which means that the defender will not gain from deviating from equilibrium strategies. Moreover, at high discount factor values, i.e, $\beta > 0.75$, the percentage increase is higher than that at lower $\beta$ values in all states. The percentage increase ranges from $5\%$ to about $40\%$ at $\beta = 0.75$ depending on the state, and it can reach values between $20\%$ and above $40\%$ at $\beta > 0.95$. This is due to the fact that, at higher $\beta$ values, future state transitions have higher impact on calculating equilibrium strategies and the defender considers more state changes in the future. This makes equilibrium strategies differ more from equal probabilities. For other $\beta$ values, the percentage increase depends on how different the equilibrium strategy from the equal allocation scheme.

\begin{figure}[t]
  \centerline{\includegraphics[width=7cm]{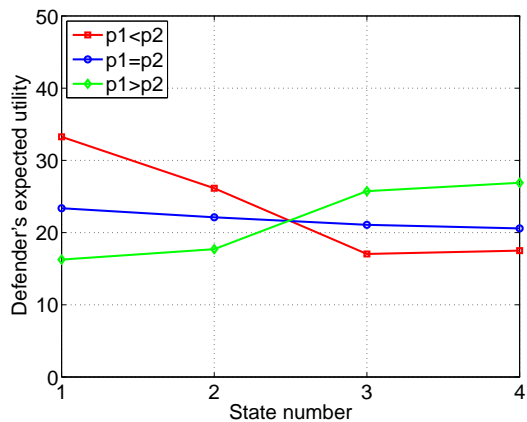}}
  \caption{The defender's expected utility in each state for different techniques power combinations.}\label{fig:3}
  \vspace{-0.45cm}
\end{figure}

In Fig.~\ref{fig:3} we study the effect of changing the power on the defender's expected utility at equilibrium. We study three cases, first when the power required for technique $1$ is less than the power required for technique $2$, similar to the previous experiments. Then, we study the cases in which they are equal and in which technique $1$ requires more power than technique $2$. Here, we set $\beta=0.75$. From Fig.~\ref{fig:3}, we can see that, when the first technique's power is less than the second one, the defender gets higher reward at states $s_1$ and $s_2$ than at states $s_3$ and $s_4$. This stems from the fact that, at states $s_1$ and $s_2$, the defender begins the game using technique $1$ (lower power) thus getting a higher reward. A similar result can be seen when the defender gets a higher reward at states $s_3$ and $s_4$ when the technique used in these states needs less power. When the two techniques use the same power, we can see that the defender's expected utility is almost the same for all states. Fig.~\ref{fig:3} clearly shows the effect of first state parameters on the expected utility.

\begin{figure}[t]
  \centerline{\includegraphics[width=7cm]{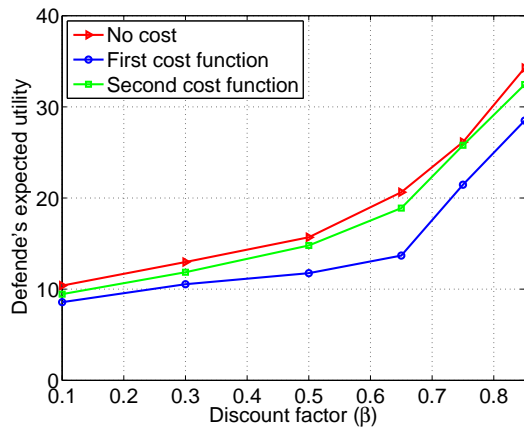}}
  \caption{The defender's expected utility in state $s_2$ against discount factor $\beta$ for different cost functions.}\label{fig:4}
    \vspace{-0.45cm}
\end{figure}

In Fig.~\ref{fig:4}, we study the effect of adding cost to the defender's expected utility. We calculate the expected utility at different discount factor values at state $s_2$. Clearly, the expected utility is higher when no cost is applied. When applying cost function $C(q,n)=q \cdot \ln(n+1)$, the utility is barely reduced. When applying the cost function $C(q,n)=q \cdot n$, we notice a significant decrease in the defender's expected utility. In our problem, the cost function $C(q,n)=q \cdot n$ will be more suitable as the other cost function does not show a significant change.

\section{Conclusions}\label{sec:conclusion}

In this paper, we have studied the use of MTD in a wireless network security problem. We have formulated the problem using a non-zero sum stochastic game theory model in which the defender controls state transition. The next state is determined only by defender's actions which is suitable for MTD cases where the defender want to change system parameter's before the attacker can reveal them. This property of the game ensures that the game will always have an equilibrium point. We have provided the mathematical model for deriving an equilibrium in such games. We then provided a novel way to define cost in MTD systems that depends on the number of consecutive changes in system parameters. We have shown two different functions to define cost and have proved that the game will still have equilibrium. Simulation results have shown that this model helps the defender to get higher expected utility in all system state than the case of assigning equal probabilities over different actions.

\def\baselinestretch{0.8}
\bibliographystyle{IEEEtran}
\bibliography{references}

\end{document}